\documentclass[a4paper,USenglish]{lipics-v2021}
\nolinenumbers
\usepackage{Style-Files/preamble}

\algdef{SN}[EVENT]{Event}{EndEvent}[1]{\textbf{upon}\ #1\ \algorithmicdo}{}
\newcommand{\leaderless}{Fast\xspace}
\newcommand{\leaderlesscaps}{Fast\xspace}

\authorrunning{N. Giridharan, M. Nguyen, I. Abraham, A. Clement, N. Crooks, P. Sutra}

\title{Optimality and Trade-offs in \leaderlesscaps{} BFT SMR (Extended Version)}

\author{Neil Giridharan}{University of California, Berkeley, CA, USA}{giridhn@berkeley.edu}{}{}
\author{Minh Tung Nguyen}{Institut Polytechnique de Paris, Palaiseau, France}{minh-tung.nguyen@telecom-sudparis.eu}{}{}
\author{Ittai Abraham}{a16z Crypto Research, New York City, NY, USA}{ittaia@gmail.com}{}{}
\author{Allen Clement}{Subzero Labs, Zurich, Switzerland}{aclement@gmail.com}{}{}
\author{Natacha Crooks}{University of California, Berkeley, CA, USA}{ncrooks@berkeley.edu}{}{}
\author{Pierre Sutra}{Institut Polytechnique de Paris, Palaiseau, France}{pierre.sutra@telecom-sudparis.eu}{}{}

\EventEditors{Ioannis Chatzigiannakis, Andrea Vitaletti, Keren Censor-Hillel, and William K. Moses Jr.}
\EventNoEds{4}
\EventLongTitle{40th International Symposium on Distributed Computing (DISC 2026)}
\EventShortTitle{DISC 2026}
\EventAcronym{DISC}
\EventYear{2026}
\EventDate{November 9--13, 2026}
\EventLocation{Rome, Italy}
\EventLogo{}
\SeriesVolume{397}
\ArticleNo{6}
\Copyright{Neil Giridharan, Minh Tung Nguyen, Ittai Abraham, Allen Clement, Natacha Crooks, and Pierre Sutra}

\begin{document}
\maketitle

\begin{abstract}
\leaderlesscaps{} state-machine replication (SMR) protocols in the crash-fault setting have attracted significant interest in both academia and industry. This interest stems from their advantages over leader-based protocols, including low execution latency for non-conflicting commands, high throughput, improved availability, and increased fairness. While \leaderless SMR is well studied in the crash-fault setting, the Byzantine fault-tolerant (BFT) setting remains much less understood, with only a small number of existing \leaderless BFT protocols. In this paper, we present tight upper and lower bounds on the replication factor required for \leaderless BFT SMR. We also present a suboptimal protocol that illustrates a trade-off between replication factor and recovery efficiency.
\end{abstract}

\section{Introduction}
\label{sec:introduction}


State machine replication (SMR) is a foundational technique for building fault-tolerant-services, ranging from replicated databases, coordination services to blockchains. SMR guarantees that all replicas execute the same sequence of commands such that they behave as a single highly available state machine.

Most widely deployed SMR protocols, including Raft~\cite{raft} and Multi-Paxos~\cite{multi-paxos}, are \emph{leader-based}. There is a distinguished replica that is responsible for assigning commands to positions in a global log. This design is simple and effective, but adds additional latency in two ways. First, every command must be routed through, and ordered by the leader. Second, in protocols such as Multi-Paxos, commands are assigned to predefined log slots. As a result, a command may be forced to wait for earlier slots to commit before it can execute, even when those earlier commands do not conflict with it.

\leaderlesscaps{} SMR protocols~\cite{epaxos, leaderless-smr, atlas, bipartisan-paxos} are designed to overcome the performance limitations inherent in traditional leader-based protocols. These protocols move away from a model where a single leader must totally order every command and instead exploit the observation that many commands commute. To ensure linearizability~\cite{linearizability}, the gold standard for SMR~\cite{smr}, replicas need only agree on a consistent order for conflicting commands; non-conflicting commands can safely be executed in different orders~\cite{generalized-paxos}.  Rather than forcing every command into a fixed slot within a totally ordered log, it is sufficient to reach agreement on the specific dependencies (ordering constraints) of each command. 

This shift in perspective allows \leaderless SMR protocols to execute non-conflicting commands with low latency. Specifically, the defining property of \leaderless SMR protocols is that they can execute such commands in at most $2\Delta$ time (two message delays) when the network is synchronous.
\leaderlesscaps{} SMR protocols also provide three additional practical benefits.
1) High Throughput: By avoiding the single-leader bottleneck, the system spreads load across all replicas, maximizing aggregate bandwidth.
2) Improved Availability: Unlike leader-based systems that stall during a leader election, leaderless protocols continue to commit commands as long as the number of faults remains below the threshold $f$.
3) Increased Fairness: Because ordering decisions are decentralized, no single replica can order commands arbitrarily. This inherently mitigates the risks of biased ordering often seen in leader-based protocols.

\leaderlesscaps{} SMR protocols have gathered significant interest in both industry and academic settings. Protocols such as EPaxos~\cite{epaxos}, Atlas~\cite{atlas}, Caesar~\cite{caesar}, Accord~\cite{accord}, and BPaxos~\cite{bpaxos} confirm the substantial benefits that the approach enables, especially for workloads with low conflict rates. In fact, Accord is now being deployed in Apple's Cassandra~\cite{cassandra} system. Beyond these practical benefits, the crash-fault setting is also relatively well understood: prior academic work has formalized results for latency and resilience of these systems~\cite{leaderless-smr,two-step,epaxos-star}.


While crash-fault leaderless protocols have led to a substantial body of work and practical deployments, we lack a comparable understanding of what is achievable when replicas may behave Byzantine. In particular, it is not clear which benefits of \leaderless SMR can be preserved, what is the necessary replication factor, or what new mechanisms are required to maintain safety. The difficulty of the Byzantine setting is illustrated by prior attempts to adapt \leaderless SMR techniques to Byzantine faults, which were later found to have correctness issues~\cite{ezbft,ezbft-revisited}. This gap motivates our study of \leaderless BFT SMR.

This paper initiates a systematic study of \leaderless BFT SMR protocols and identifies a fundamental trade-off between optimal resilience and recovery efficiency. We highlight three key results.

\begin{itemize}
    \item Our main result is an upper bound showing that \leaderless BFT SMR is achievable with the optimal replication factor of $n \geq 5f+1$ replicas (Algorithm~\ref{alg:replica-part-1}). Prior \leaderless BFT SMR protocols with a replication factor of $n\geq 5f+1$ either required a fault-free fast path~\cite{leaderless-smr,basil} or tightly synchronized clocks~\cite{flutter,aspen}. In contrast, our protocol requires neither: it tolerates up to $f$ failures while preserving fast path execution and does not require tightly synchronized clocks. The key technical ingredient is a new recovery procedure that safely recovers fast-path decisions while ensuring that conflicting commands are always executed in a consistent order. This allows the protocol to simultaneously achieve an optimal replication factor and fast execution latency for non-conflicting commands.
      
    \item We then prove a matching lower bound, showing that this replication factor is tight (Theorem~\ref{thm:leaderless-lower-bound}). Specifically, we show that no \leaderless BFT SMR protocol exists for $n \leq 5f$. The lower bound follows from the requirement that \leaderless protocols execute non-conflicting commands in two message delays (including client communication).

    \item Finally, we present a simpler but suboptimal \leaderless BFT SMR protocol for $n \geq 7f+1$ replicas (Algorithm~\ref{alg:replica-sevenf-plusone}). This protocol illustrates the trade-off between replication factor and recovery efficiency: with larger quorums, recovery no longer requires our new mechanism and can instead rely on a much simpler procedure that is more efficient.
\end{itemize}


\section{Model and Definitions}
\label{sec:model}

\subsection{System Model}
We consider a system of $n$ replicas $\Pi=\{1,\ldots,n\}$. Replicas communicate
over a reliable, authenticated, point-to-point network. A replica is
\emph{correct} if it follows the protocol and is \emph{Byzantine} otherwise.
Byzantine replicas may behave arbitrarily: they may equivocate, omit messages,
send malformed messages, or coordinate their behavior with other Byzantine
replicas. As is standard, we assume that at most $f$ replicas are Byzantine.

We consider a static adversary: the adversary chooses the faulty replicas before the execution begins. The adversary may coordinate Byzantine replicas and schedule message deliveries
subject to the network assumptions below, but cannot break standard
cryptographic primitives.

We adopt the partially synchronous model. There is a known message-delay bound
$\Delta$ and an unknown Global Stabilization Time, denoted GST, such that after
GST every message sent between correct replicas is delivered within $\Delta$
time. Before GST, messages may be delayed arbitrarily, but the network does not corrupt messages. 

We assume standard digital signatures and a public-key infrastructure. We write
$\langle m\rangle_k$ to denote a message $m$ signed by replica $k$. A message is
\emph{well-formed} if all signatures it contains are valid and the message
satisfies the syntactic requirements of the protocol. Correct replicas ignore
malformed messages.

\subsection{Definitions}

\par{\textbf{BFT SMR.}} In Byzantine fault-tolerant state-machine replication, clients \textbf{submit} commands to be executed by replicas, and replicas \textbf{execute} commands by applying them to the state machine~\cite{smr}. Without loss of generality, we assume that commands are unique and have unique, collision-resistant identifiers.

More precisely, there is a public verification predicate $\mathsf{VerifyId}(id,x)$ that determines whether $id$ is a valid identifier for command payload $x$. We also assume that identifiers are unforgeable, so only the client that submitted command payload $x$ can produce an identifier $id$ such that $\mathsf{VerifyId}(id,x)$ holds. We also assume collision resistance: for any command payload $x$ and identifier $id$ such that $\mathsf{VerifyId}(id,x)$ holds, no adversary can find a distinct command payload $y\neq x$ such that $\mathsf{VerifyId}(id,y)$ also holds. A standard way of generating identifiers that satisfies this definition is to use a signed digest of the command payload.

A command is valid if it satisfies the public verification predicate and has a valid client signature. A command $x$ conflicts with command $y$ (we write $x\bowtie y$) if the relative order of executing these two commands affects the final system state. Two commands that do not conflict commute ($x\not\bowtie y$). The $\mathsf{noop}$ command is a special command that does not conflict with any other command. Database systems, for instance, say that transactions conflict if both of them access the same key and one of the accesses is a write. As observed in prior work~\cite{generalized-paxos, byzantine-generalized-paxos, leaderless-smr}, BFT SMR only needs to order conflicting commands to satisfy linearizability. We provide the formal definition of BFT SMR below.


\begin{definition}[BFT SMR]
\label{def:generalized-leaderless-smr}
A protocol implements Byzantine fault-tolerant state machine replication if it satisfies the following properties when there are at most $f$ failures.
\begin{description}

    \item[Safety.]
    Correct replicas execute conflicting commands in the same order.

    \item[Liveness.]
    In an execution with a finite number of submitted commands, if a correct replica executes command $c$ or a correct client submits command $c$, then all correct replicas eventually execute $c$. 



    \item[External validity.]
    Every command executed by a correct replica was submitted by some client.
    
\end{description}
\end{definition}

Note that we have a stronger assumption on liveness that guarantees termination only when there are a finite number of commands. This property matches the liveness properties that \leaderless protocols like EPaxos~\cite{epaxos, epaxos-star} provide.

\par{\textbf{\leaderlesscaps{} BFT SMR.}} The defining characteristic of \leaderless BFT SMR protocols is that they can execute commands in the minimal amount of latency (within $2\Delta$) in \textit{synchronous conflict-free} execution. Informally, a synchronous conflict-free execution is one in which it is after GST and there are no concurrently submitted conflicting commands. This stringent requirement is not easily met. Leader-based protocols~\cite{pbft,hotstuff,autobahn} cannot satisfy it as routing client requests to the leader necessarily adds an additional message delay. Similarly, in protocols with a predefined slot order such as Mir-BFT~\cite{mir-bft} a command in slot $i$ may have to wait for all lower-numbered slots to commit before it can execute, even when those slots committed non-conflicting commands. We formally define this performance constraint during conflict-free executions, adapting the definitions from~\cite{epaxos-star, two-step} to the Byzantine setting.


\begin{definition}[Conflict-freedom]
\label{def:conflict-free}
A command $c$ \emph{precedes} a command $c'$ in an execution if, at the time
$c'$ is submitted by a client, all correct replicas have already executed $c$.
Commands $c$ and $c'$ are \emph{concurrent} if neither precedes the other. A
command $c$ is \emph{conflict-free} if it does not conflict with any concurrent
command.
\end{definition}

\begin{definition}[Optimal execution latency.]
After GST, if a correct client submits a valid conflict-free command $x$ at time $t$, then every correct replica executes $x$ by time $t+2\Delta$.
\end{definition}



\begin{definition}[\leaderlesscaps{} BFT SMR]
    A \leaderless BFT SMR protocol satisfies all properties of BFT SMR and optimal execution latency.
\end{definition}

\section{\leaderlesscaps{} SMR Background}

In this section, we give an overview of how the crash-fault tolerant (CFT) family of \leaderless SMR protocols works. We use a similar framework to construct our upper bound protocol for the BFT setting.

Because non-conflicting commands need not be ordered, \leaderless SMR protocols can run many single-shot agreement instances in parallel. Rather forming a totally ordered log of instances, the \leaderless SMR protocol constructs a partial order over instances (represented by command identifiers) such that any two conflicting commands need to be consistently ordered across replicas. This is achieved by having each instance agree on a tuple consisting of two components: 1) the command (payload) itself and 2) the dependencies of the command, which is the set of conflicting commands (specifically command identifiers) that this operation needs to be ordered after.

The collection of committed commands defines a dependency graph. Each committed command is a vertex, and there is an edge from each command to every command in its committed dependency set. Since replicas run an agreement instance for every command, all correct replicas agree on the dependency graph. Execution is then deterministic. A replica waits until all transitive dependencies of a command have committed, computes the strongly connected components of the resulting
dependency graph, orders those components in topological order, and executes commands inside each component by increasing command identifiers. Because the linearization algorithm is deterministic, all correct replicas execute the same
commands in the same order. The execution algorithm pseudocode is presented in Algorithm~\ref{alg:execute-committed}.

To maintain safety, \leaderless SMR protocols need to satisfy two key invariants~\cite{leaderless-smr, epaxos-star}, which ensure that conflicting commands are always executed in the same order. The first invariant ensures that replicas agree on the same dependency graph. The second variant ensures that committed conflicting commands always have an edge in the dependency graph.

\begin{theorem}[Invariant 1: Agreement]
\label{thm:agreement-invariant}
For any command identifier $id$, if both $(c,D)$ and $(c',D')$ are committed for $id$, then $c=c'$ and $D=D'$.
\end{theorem}

\begin{theorem}[Invariant 2: Visibility]
\label{thm:visibility-invariant}
    If $(c,D)$ is committed and $(c',D')$ is committed for command identifiers $id$ and $id'$ respectively, where $id\neq id'$ and $c\bowtie c'$, then $id\in D'$ or $id'\in D$.
\end{theorem}

\leaderlesscaps{} SMR protocols consistent of three components: a fast path that is guaranteed to succeed for conflict-free commands after GST. A recovery path, which guarantees that if the fast path fails, some value is eventually committed. Finally, an execution engine that executes commands

To achieve optimal execution latency, the agreement protocol includes a fast path that succeeds when a quorum of replicas have the same input value. In synchronous executions, conflict-free commands are guaranteed to take this fast path. Moreover, because all replicas have already executed the command’s committed dependencies, the command can be executed immediately once it commits.

If the fast path does not succeed, either due to asynchrony or concurrent conflicting commands, the recovery path guarantees that some value is eventually committed. Recovery is view-based: each view has a coordinator responsible for driving agreement. To preserve safety, the coordinator must propose a value that is consistent with any value that may have committed on the fast path or in an earlier view, while also ensuring that conflicting commands cannot miss each other.

\begin{algorithm}[H]
\caption{Execution algorithm}
\label{alg:execute-committed}
\begin{algorithmic}[1]
\State $\mathit{executed}\gets \emptyset$
\While{\textbf{true}}
    \State Let $G \subseteq \mathit{dep}$ be the largest subgraph such that
    $\forall id \in G,$ $id$ is committed and all of its dependencies are in $G$
    \For{$C \in \mathsf{SCC}(G)$ in topological order}
        \For{$id \in C$ in increasing rank order of identifiers}
            \If{$id \notin \mathit{executed}$}
                \State \textbf{execute} command committed for $id$
            \EndIf
            \State $\mathit{executed} \gets \mathit{executed} \cup \{id\}$
        \EndFor
    \EndFor
\EndWhile
\end{algorithmic}
\end{algorithm}

\section{$n\geq 5f+1$ Protocol Overview}
\label{sec:overview}


We first show that there exists a protocol that solves \leaderless BFT SMR for $n\geq 5f+1$ replicas. The protocol consists of a series of instances, each identified by a request identifier $id$. A client request \msg{Req}{id,x} associates instance $id$ with command payload $x$. The purpose of instance $id$ is to commit a value $(x',D)$, where $x'$ is either a command payload or $\mathsf{noop}$ and $D$ is a set of dependencies that represent ordering constraints. Once a replica has committed all of the transitive dependencies for the instance, it can run the execution algorithm (Algorithm~\ref{alg:execute-committed}) to execute the command payload.

Each instance consists of two commit paths: \one a fast path that commits in two message delays after GST when there are no concurrent conflicts and \two a recovery path that makes progress when conflicts do occur.



\subsection{Fast Path}

To satisfy optimal execution latency, the protocol must commit conflict-free commands in two message delays after GST. Since the client already uses one message delay to send the command to replicas, the protocol has only one additional message delay before it must commit. The fast path therefore consists of replicas broadcasting a vote for their initial value. Since the fast path must tolerate up to $f$ faulty replicas, a replica must commit after receiving $n-f$ matching votes for the same value. If replicas observe non-matching votes, they enter the recovery path.

\subsection{Recovery Path}

To guarantee liveness, the recovery path must also make progress using responses from only $n-f$ replicas. At the same time, safety requires recovery to preserve any value that may already have been committed on the fast path. 

If some replica committed a value from $n-f$ fast votes, then by quorum intersection a later recovery quorum is guaranteed to observe only $n-3f$ of those votes. Thus, for safety, recovery must treat those $n-3f$ votes as sufficient evidence that a value may have been committed on the fast path. This is because up to $f$ of the $n-f$ votes may be from Byzantine replicas, while asynchrony may cause the $n-f$ quorum to miss $f$ correct replicas.





\par \textbf{Main challenge.} The main challenge for recovery comes from satisfiying Invariant~\ref{thm:visibility-invariant} (Visibility) when it is possible for replicas to receive $n-3f$ fast votes for $(x,\emptyset)$ and $n-3f$ fast votes for $(y,\emptyset)$ ($x\bowtie y$), indicating that both could have gone fast path.

Consider for example two conflicting commands $x$ and $y$, and two disjoint sets of replicas $A$ and $B$, each of size $n-3f$. Replicas in $A$ receive $x$ before $y$ and send a fast vote for $x$ with empty dependencies. Replicas in $B$ receive $y$ before $x$ and send a fast vote for $y$ with empty dependencies. Even though neither fast path actually succeeds, recovery for $x$ may observe $n-3f$ fast votes for $x$ with empty dependencies, while recovery for $y$ may observe $n-3f$ fast votes for $y$ with empty dependencies. Since each value appears consistent with a possible fast path commit, both recovery paths may try to commit their respective values. However, committing both values would violate Invariant~\ref{thm:visibility-invariant} (Visibility): $x$ and $y$ would be committed with neither being in the dependency set of the other. The problem is that recovery cannot distinguish this execution from one in which the observed $n-3f$ votes are evidence of a real fast path commit.

\par \textbf{Key Idea.}
Our key idea is to make recovery collaborative. While it may not be possible to determine whether a particular command committed on the fast path in isolation, it may be possible to do so after learning about the outcome of other commands! Our main contribution is a deadlock-free recovery path that waits to learn the outcomes of precisely those commands whose decisions could affect the command being recovered.

To make recovery deadlock-free, our protocol imposes a fixed total order on instance identifiers. Recovery is only allowed to wait downward in this order: a higher-ranked instance may wait for a lower-ranked instance, but a lower-ranked instance never waits for a higher-ranked one. This makes the wait-for graph acyclic. The lowest-ranked pending recovery cannot be blocked by another pending recovery, and once it completes, higher-ranked recoveries can continue.

\begin{itemize}
    \item To determine which lower-ranked instances it must wait for, the recovery path begins with a validation phase. Replicas collect responses from $n-f$ replicas, each reporting any fast vote it sent for a lower-ranked conflicting instance. If any response contains such a vote, recovery waits for that instance to commit before proceeding. Any lower-ranked instance that is not reported during the validation phase can be safely ignored. At most $f$ correct replicas lie outside an $n-f$ quorum, and at most $f$ Byzantine replicas within the quorum could lie about their votes. Thus, at most $2f$ replicas could have sent fast votes for a conflicting value, which is fewer than the $n-3f$ threshold required to recover a value.

    \item Recovery therefore has two rules. First, it may only wait for lower-ranked instances, which prevents deadlock. Second, it must wait for every conflicting lower-ranked instance discovered during the validation phase, which preserves safety. After those lower-ranked instances commit, recovery chooses a value that is consistent with their committed outcomes. If the value supported by its own fast-vote evidence is still safe, it recovers that value. Otherwise, it recovers $\mathsf{noop}$, which is a special command that does not conflict with any other command.
\end{itemize}


\section{$n\geq 5f+1$ Protocol}
\label{sec:n-ge-5f-plus-1-protocol}

\providecommand{\AlgLine}[2]{Algorithm~\ref{#1}, line~\ref{#2}}
\providecommand{\AlgLines}[3]{Algorithm~\ref{#1}, lines~\ref{#2}--\ref{#3}}
Having sketched out an overview of the key techniques of the protocol, we now provide a detailed protocol specification for our \leaderless BFT state-machine replication
protocol with $n\geq 5f+1$ replicas. The protocol is described in Algorithm~\ref{alg:replica-part-1}. We assume that all messages are signed and that replicas verify both message signatures and message validity. For clarity, the pseudocode omits validation checks that ensure messages conform to the protocol specification. 



\par \textbf{Fast path.} Clients issue requests to replicas by sending a \msg{Req}{id,x} message, where $id$ is command identifier for a command payload $x$. Clients generate identifiers from taking a signed digest of ($c_{id}$, $seq$, $x$), where $c_{id}$ is a unique client id, $seq$ is a monontically increasing sequence number, and $x$ is the command payload.
When a replica receives a client request \msg{Req}{id,x}, either directly from
the client or forwarded by another replica it first records the command $x$ in $known$
and forwards the request to all replicas if this is the first time it has seen
the command
(\AlgLines{alg:replica-part-1}{line:p1-record-req}{line:p1-forward-req}).
If the replica has not already sent a fast vote or entered recovery
(\AlgLine{alg:replica-part-1}{line:p1-req-exit}), it computes the dependency set from conflicting commands in its local $known$ state
(\AlgLine{alg:replica-part-1}{line:p1-compute-deps}), records its fast vote
(\AlgLine{alg:replica-part-1}{line:p1-store-fastvote}), and broadcasts the
\textsc{FastVote} to all replicas
(\AlgLine{alg:replica-part-1}{line:p1-send-fastvote}).

A replica commits on the fast path after receiving at least $n-f$
\textsc{FastVote} messages with matching dependencies. It stores the resulting commit certificate, forwards the commit
certificate to all replicas, and commits the value
(\AlgLines{alg:replica-part-1}{line:p1-fast-commit-store}{line:p1-fast-commit}).
Importantly, replicas continue monitoring for fast votes and can commit on the fast path even if they have already entered the recovery path, and cannot send fast votes themselves. This is necessary for optimal execution latency: even in a synchronous conflict-free execution, a correct replica may enter recovery before it observes enough fast votes to commit on the fast path.

\par \textbf{Entering recovery.}
If a replica receives $n-f$ \textsc{FastVote} messages with different dependencies, it stores this set of messages and forwards this evidence to other replicas (\AlgLines{alg:replica-part-1}{line:p1-forward-nonmatch}{line:p1-store-nonmatch}). Forwarding these fast votes ensures that all correct replicas eventually enter recovery. A replica waits for $\Delta$ time before entering recovery (\AlgLine{alg:replica-part-1}{line:p1-wait-fast-path}). In synchronous conflict-free executions, this delay gives a client request enough time to reach the replica and prevents Byzantine replicas from forcing the replica into recovery before it has a chance to commit on the fast path. After this waiting period, replicas enter recovery.

Recovery consists of a sequence of views, where each view has a dedicated coordinator, similar to classical leader-based BFT protocols. Each view consists of four phases: a status phase, a validation phase, a proposal phase, and a voting phase. The status phase determines whether some value might already have committed in a previous recovery view or on the fast path. The validation phase determines whether the coordinator must wait
for lower-ranked conflicting instances before safely proposing a value. Once a safe value to propose has been identified, recovery runs standard proposal and voting phases, akin to traditional leader-based BFT protocols.



When entering any recovery view, a replica updates its view and stops replying to messages in earlier views
(\AlgLine{alg:replica-part-1}{line:p1-record-view}). After this point, the
replica sends no fast vote or messages from earlier views for that instance. The
replica starts a view timer
(\AlgLine{alg:replica-part-1}{line:p1-start-timer}) and sends the coordinator a
\textsc{Status} message containing the replica's fast vote and highest recovery
vote for the instance
(\AlgLine{alg:replica-part-1}{line:p1-send-status}).

\par \textbf{Status phase.}
The coordinator of the view waits for $n-f$ valid \textsc{Status} messages in the current view. If
the collected \textsc{Status} messages contain $n-3f$ recovery votes for the same value,
then that value might already have committed on the recovery path in an earlier
view. Since these votes were sent after the validation phase, the coordinator proposes it directly
(\AlgLine{alg:replica-part-1}{line:p2-propose-prior-val}). Otherwise, if the
\textsc{Status} messages contain $n-3f$ \textsc{FastVote} messages for the same value $(x,deps)$, then
that value might have committed on the fast path. In this case, the coordinator
starts validation by broadcasting a \textsc{Validate} message carrying the value and
the status certificate
(\AlgLine{alg:replica-part-1}{line:p2-send-validate}). If neither condition
holds, then the coordinator computes the union of the dependencies contained within the $n-f$ \textsc{FastVote} messages it used to enter view $1$, and proposes this value
(\AlgLines{alg:replica-part-1}{line:p2-compute-union}{line:p2-propose-union}).

\par \textbf{Validation phase.}
Validation is used only when the status phase finds a value that might have
committed on the fast path. Upon receiving a valid \textsc{Validate} message, a
replica first checks that it is still in the indicated view
(\AlgLine{alg:replica-part-2}{line:p2-validate-viewcheck}). It then records the
command $x$ in $known$
(\AlgLine{alg:replica-part-2}{line:p2-validate-known}), ensuring that any future
conflicting command learned by this replica must account for $x$ in its
dependency computation.

The replica next computes a set $wait$ of lower-ranked instances that the coordinator
must wait for before safely proposing $(x,deps)$
(\AlgLine{alg:replica-part-2}{line:p2-wait-init}). For every lower-ranked
instance $id'<id$ for which the replica sent a \textsc{FastVote} message, the replica checks
whether the earlier command $y$ for instance $id'$, conflicts with $x$, whether $id$ is absent from
the dependency set of $y$, and whether $id'$ is absent from $deps$
(\AlgLines{alg:replica-part-2}{line:p2-check-lower-instance}{line:p2-add-wait}). Violating a single one of these conditions is sufficient to guarantee that Invariant~\ref{thm:visibility-invariant} holds: either the commands commute (precondition does not hold) or at least one command has a dependency on the other. If all of these conditions hold, then committing both values would violate Invariant~\ref{thm:visibility-invariant}, so the replica includes the corresponding client
request in $wait$. The replica returns this information to the coordinator in a
\textsc{ValidateOk} message
(\AlgLine{alg:replica-part-2}{line:p2-send-validateok}).

The coordinator waits for $n-f$ valid \textsc{ValidateOk} messages, unions their $wait$
sets, and initializes the set of commit certificates it will collect
(\AlgLines{alg:replica-part-2}{line:p2-all-waits}{line:p2-commitcerts-init}).
For each lower-ranked request in the union, the coordinator forwards the request to
all replicas so that the corresponding instance will eventually make progress
(\AlgLine{alg:replica-part-2}{line:p2-forward-waited-req}). The coordinator then
waits for a commit certificate for that instance
(\AlgLine{alg:replica-part-2}{line:p2-wait-commit}), records the committed value
and certificate
(\AlgLines{alg:replica-part-2}{line:p2-get-committed}{line:p2-add-commit-cert}),
and checks whether the committed value conflicts with the proposal in
a way that violates safety
(\AlgLine{alg:replica-part-2}{line:p2-conflict-check}). If such a conflicting
commit is found, then the proposal could not have safely committed on the
fast path, and the coordinator proposes ($\mathsf{noop}$, $\emptyset$) instead
(\AlgLines{alg:replica-part-2}{line:p2-propose-noop-conflict}{line:p2-exit-after-noop}).
If all waited-on instances commit values that do not violate safety, then the
coordinator proposes the original value $(x,deps)$ together with the collected
evidence
(\AlgLine{alg:replica-part-2}{line:p2-propose-original}).

\par \textbf{Proposal and voting.}
A proposal contains the value for the current view together with enough evidence
for replicas to verify that the coordinator followed the status and validation rules. A replica verifies that the coordinator picked the correct value to propose based on the given evidence (for clarity we omit this check from the pseudocode).
When a replica receives a valid \textsc{Propose} message, it ignores proposals for
old views
(\AlgLines{alg:replica-part-2}{line:p2-propose-viewcheck}{line:p2-propose-view-exit})
and refuses to vote twice in the same view
(\AlgLines{alg:replica-part-2}{line:p2-propose-dupcheck}{line:p2-propose-dup-exit}).
If the proposed value is not $\mathsf{noop}$, the replica records the command in
$known$
(\AlgLines{alg:replica-part-2}{line:p2-propose-record-known-check}{line:p2-propose-record-known}).
It then stores its recovery vote and broadcasts a \textsc{Vote} message to all
replicas
(\AlgLines{alg:replica-part-2}{line:p2-record-vote}{line:p2-send-vote}).


\par \textbf{Commit.}
A recovery-path commit occurs when a replica receives $n-f$ matching
\textsc{Vote} messages in the same view. The replica stores the resulting commit
certificate, forwards it to all replicas, and commits the value
(\AlgLines{alg:replica-part-2}{line:p2-recovery-commit-store}{line:p2-recovery-commit}). When a replica commits either on the fast path or recovery path, it forwards all the \textsc{Req} messages in the committed dependency set, so that eventually these dependencies will also be committed (\AlgLines{alg:replica-part-2}{line:p1-forward-commit-req}{line:p1-forward-commit-req-end}).

\par \textbf{View change.} If progress stalls in a recovery view, the view timer expires and the replica
broadcasts a \textsc{ViewChange} message
(\AlgLine{alg:replica-part-2}{line:p1-send-viewchange}). Once a replica receives
$f+1$ such messages for a view, it enters the next view and forwards the
view-change evidence to all replicas
(\AlgLines{alg:replica-part-2}{line:p1-enter-next-view}{line:p1-forward-viewchange}).
This mechanism ensures view synchronization among correct replicas.

\par \textbf{Execution.} When a replica commits a command and set of dependencies, it must wait until the set of transitive dependencies have been committed before its ready to execute. Execution then follows the same algorithm as in the EPaxos family of algorithms~\cite{epaxos, epaxos-star}, which is stated in Algorithm~\ref{alg:execute-committed} for completeness.

\par \textbf{Resubmission.} A client may resubmit a command payload in case the previous instance committed $\mathsf{noop}$, in which case the client sends a new request together with a commit certificate of the result of the previous instance. Replicas check the validity of this certificate before voting.




\begin{algorithm}[htbp]
\caption{Replica $i$ protocol for $n\geq 5f+1$}
\label{alg:replica-part-1}
\begin{algorithmic}[1]

\State \label{line:p1-known}$known\gets \emptyset$\Comment{Set of received commands, used to compute dependency sets}
\State $fastvotes \gets \emptyset$\Comment{Fast vote sent per instance}
\State $fv\gets \emptyset$\Comment{Set of fast votes received per instance}
\State $votes \gets \emptyset$\Comment{Highest vote (by view) sent}
\State $commits \gets \emptyset$\Comment{Commit certificates}
\State \label{line:p1-views}$views \gets \emptyset$\Comment{Current view for each instance}

\medskip
\Event{receiving $r\gets $\msg{Req}{id, x} from client or replica}
    \If{$r\notin known$}
        \State \label{line:p1-record-req}$known\gets known\cup\{r\}$
        \State \label{line:p1-forward-req}\textbf{forward} \msg{Req}{id,x} to all
    \EndIf

    \If{$views[id]\geq 1$ or $fastvotes[id]\neq \bot$}
        \State \label{line:p1-req-exit}\textbf{return}
    \EndIf
    
    \State \label{line:p1-compute-deps}$deps\gets \{(id',$ \msg{Req}{id',y} $)\mid \exists$ \msg{Req}{id',y} $\in known \text{ such that } y \bowtie x\}$
    \State \label{line:p1-store-fastvote}$fastvotes[id]\gets (x,deps)$
    \State \label{line:p1-send-fastvote}\textbf{send} \msg{FastVote}{id,(x,deps)} to all
\EndEvent

\medskip
\Event{receiving $C\gets n-f$ matching \msg{FastVote}{id, (x, deps)}}\Comment{Can fast commit in higher views}
    \State \label{line:p1-fast-commit-store}$commits[id] \gets C$
    \State \label{line:p1-fast-commit-forward}\textbf{forward} $C$ to all
    \State \label{line:p1-fast-commit}\textbf{commit} $(id,(x, deps))$
\EndEvent

\medskip
\Event{receiving $F\gets n-f$ non-matching \msg{FastVote}{id,\cdot}}
    \State \label{line:p1-store-nonmatch} $fv[id] \gets F$
    \State \label{line:p1-forward-nonmatch}\textbf{forward} $F$ to all
    \State \label{line:p1-wait-fast-path} \textbf{wait} $\Delta$\Comment{Wait for possible fast path before entering view 1}
    \State \label{line:p1-enter-view1}\textbf{enter} view $1$ for instance $id$
\EndEvent

\medskip
\Event{entering view $v$ for instance $id$}
    \If{$views[id]\neq \bot$ and $views[id]\geq v$}\label{line:p1-stale-view-check}
        \State \textbf{return}
    \EndIf
    \State \label{line:p1-record-view}$views[id]\gets v$
        \Comment{After this point, replica $i$ sends no fastvote/earlier view vote for $id$}
    \State \label{line:p1-start-timer}\textbf{start} timer with duration $6\Delta$ for view $v$ in instance $id$
    \State \label{line:p1-send-status}\textbf{send} \msg{Status}{id, v, fastvotes[id], votes[id]} to $\Call{Coordinator}{id, v}$
\EndEvent

\medskip
\Event{receiving $S\gets n-f$ valid \msg{Status}{id, v, \cdot,\cdot} and $i=\Call{Coordinator}{id,v}$}
    \If{there exists $n-3f$ $votes[id]$ with $val$ from $S$}
        \State \label{line:p2-propose-prior-val}\textbf{send} \msg{Propose}{id,v,val,S,\bot,\bot} to all\Comment{Possibly went slow path}
    \ElsIf{there exists $n-3f$ $fastvotes[id]$ with $(x,deps)$ from $S$}
        \State \label{line:p2-send-validate}\textbf{send} \msg{Validate}{id, v, (x,deps), S} to all\Comment{Possibly went fast path}
    \Else
        \State $x\gets $ unique command in in $fv[id]$
        \State \label{line:p2-compute-union} $D\gets$ union of all dependencies from $fv[id]$
        \State \label{line:p2-propose-union} \textbf{send} \msg{Propose}{id, v, (x,D), S,\bot,\bot} to all
    \EndIf
\EndEvent

\algstore{myalg}

\end{algorithmic}
\end{algorithm}

\begin{algorithm}[htbp]
\ContinuedFloat
\caption{Replica $i$ Protocol, continued}
\label{alg:replica-part-2}
\begin{algorithmic}[1]
\algrestore{myalg}

\medskip
\Event{receiving valid \msg{Validate}{id,v,(x,deps),S}}
    \If{$views[id]\neq v$}\label{line:p2-validate-viewcheck}
        \State \textbf{exit upon}
    \EndIf

    \State \label{line:p2-validate-known}$known\gets known\cup\{$\msg{Req}{id,x}$\}$

    \State \label{line:p2-wait-init}$wait\gets \{\}$
    \For{$\forall id'<id$ such that $fastvotes[id']\neq \bot$}\label{line:p2-check-lower-instance}
        \State $(y,deps')\gets fastvotes[id']$
        \If{$y$ conflicts with $x$ and $id\notin deps'$ and $id'\notin deps$}
            \State \label{line:p2-add-wait}$wait\gets wait\cup\{$ \msg{Req}{id',y} $\}$\Comment{Client req message}
        \EndIf
    \EndFor

    \State \label{line:p2-send-validateok}\textbf{send} \msg{ValidateOk}{id,v,(x,deps),S,wait} to $\Call{Coordinator}{id,v}$
\EndEvent

\medskip
\Event{receiving $V\gets n-f$ valid \msg{ValidateOk}{id, v, (x,deps), S,wait} and $i=\Call{Coordinator}{id,v}$}
    \State \label{line:p2-all-waits}$all\_waits\gets$ union of all $wait$ sets
    \State \label{line:p2-commitcerts-init}$commit\_certs \gets \{\}$
    \For{$\forall$ \msg{Req}{id',c'} $\in all\_waits$}
        \State \label{line:p2-forward-waited-req}\textbf{send} \msg{Req}{id',c'} to all\Comment{Ensures other correct replicas vote in $id'$}
        \State \label{line:p2-wait-commit}\textbf{wait} until $commits[id']\neq \bot$ before proceeding
        \State \label{line:p2-get-committed}$val'\gets$ committed value from $commits[id']$
        \State \label{line:p2-add-commit-cert}$commit\_certs\gets commit\_certs \cup \{commits[id']\}$
        \If{$val'=(y,D')$ and $y$ conflicts with $x$
            and $id\notin D'$ and $id'\notin deps$}\label{line:p2-conflict-check}
            \State \label{line:p2-propose-noop-conflict}\textbf{send} \msg{Propose}{id,v,noop,S,V,commits[id']} to all
            \State \label{line:p2-exit-after-noop}\textbf{exit upon}
        \EndIf
    \EndFor
    
    \State \label{line:p2-propose-original}\textbf{send} \msg{Propose}{id, v, (x, deps), S, V, commit\_certs} to all\Comment{no conflicting commit found}
\EndEvent

\medskip
\Event{receiving valid \msg{Propose}{id, v, val, \cdot,\cdot,\cdot}}
    \If{$views[id]\neq v$}\label{line:p2-propose-viewcheck}
        \State \label{line:p2-propose-view-exit}\textbf{exit upon}
        \Comment{Do not vote in an old view}
    \EndIf
    \If{already voted for instance $id$ in view $v$}\label{line:p2-propose-dupcheck}
        \State \label{line:p2-propose-dup-exit}\textbf{exit upon}
    \EndIf
    \If{$val=(x,deps)$}\label{line:p2-propose-record-known-check}
        \State \label{line:p2-propose-record-known}$known\gets known\cup\{$\msg{Req}{id,x}$\}$
    \EndIf
    \State \label{line:p2-record-vote}$votes[id] \gets (v,val)$
    \State \label{line:p2-send-vote}\textbf{send} \msg{Vote}{id, v, val} to all
\EndEvent

\medskip
\Event{receiving $C\gets n-f$ \msg{Vote}{id, v, val}}
    \State \label{line:p2-recovery-commit-store}$commits[id] \gets C$
    \State \label{line:p2-recovery-commit-forward}\textbf{forward} $C$ to all
    \State \label{line:p2-recovery-commit}\textbf{commit} $(id, val)$
\EndEvent

\medskip
\Event{committing $(id,(x,D))$}
    \For{$(id,msg)\in D$}\label{line:p1-forward-commit-req}
        \State $\textbf{forward}$ $msg$ to all
    \EndFor\label{line:p1-forward-commit-req-end}
\EndEvent

\medskip
\Event{view $v$ timer for instance $id$ expiring}
    \State \label{line:p1-send-viewchange}\textbf{send} \msg{ViewChange}{id, v} to all
\EndEvent

\medskip
\Event{receiving $V\gets f+1$ \msg{ViewChange}{id, v}}
    \State \label{line:p1-enter-next-view}\textbf{enter} view $v+1$ for instance $id$
    \State \label{line:p1-forward-viewchange}\textbf{forward} $V$ to all
\EndEvent

\end{algorithmic}
\end{algorithm}

\section{Lower Bound for \leaderlesscaps{} BFT SMR}
\label{sec:lower-bound-generalized-leaderless}

\newcommand{\steps}[3]{\ensuremath{{\mathbf{#1}}^{#2}}_{#3}}
\newcommand{\run}{\ensuremath{\sigma}}
\newcommand{\indist}{\ensuremath{\thicksim}}
\newcommand{\indistinguishable}[1]{\ensuremath{\stackrel{\mathrm{#1}}{\indist}}}
\newcommand{\byz}[1]{{\color{red}{#1}}}
\newcommand{\submit}[1]{\ensuremath{\mathit{submit}(#1)}}
\newcommand{\receive}[1]{\ensuremath{\mathit{receive}(#1)}}
\newcommand{\execute}[2]{\ensuremath{\mathit{execute}_{#2}(#1)}}
\newcommand{\E}{\ensuremath{\mathcal{E}}}

We now show that the resilience bound $n \ge 5f+1$ is necessary. Thus, the protocol from the previous section is optimal with respect to resilience. The proof follows the structure of the FaB Paxos lower bound~\cite{fab}, but adapts the argument to the \leaderless BFT SMR setting. Throughout this section, we rely on the optimal execution latency, liveness, and safety properties stated in Section~\ref{sec:model}.

\par \textbf{Intuition.}
The core difficulty is that, when $n=5f$, correct replicas can be equally split on two different values. Consider an execution in any \leaderless BFT SMR protocol, in which $2f$ correct replicas observe $x$ before a conflicting command $y$, another $2f$ correct replicas observe $y$ before $x$, and the remaining $f$ replicas are Byzantine. The Byzantine replicas can equivocate, making one correct replica see $3f$ votes supporting $x$ before $y$, while another correct replica sees $3f$ votes supporting $y$ before $x$.

Now focus on the replica that sees $3f$ votes for $x$ before $y$. From its perspective, this execution is indistinguishable from another execution in which $3f$ correct replicas actually voted for $x$ before $y$, $f$ correct replicas voted for $y$ before $x$, and $f$ replicas are Byzantine. In this execution, the Byzantine replicas can behave so that some correct replica observes $4f$ votes for $x$ before $y$. Since $n-f=4f$, this is indistinguishable from an execution in which all correct replicas vote $x$ before $y$. Therefore, by the latency and liveness requirements, the protocol can be forced to commit $x$ before $y$.

By symmetry, another correct replica can be forced to commit $y$ before $x$. Thus, correct replicas may execute conflicting commands in different orders, violating generalized \leaderless BFT SMR safety.

\par \textbf{Overview.}
Our lower bound uses an indistinguishability argument to reach a contradiction. We split replicas into five disjoint groups each of size at most $f$ ($A$, $B$, $C$, $D$, and $E$), and construct five executions ($\run_1,\cdots,\run_5$). 

The construction proceeds in three stages.
First, $\run_1$ and $\run_5$ are conflict-free synchronous executions for $x$ and $y$, respectively.
By the optimal execution latency property, the correct replicas in $A$ execute $x$ by time $2\Delta$ in $\run_1$ and execute $y$ by time $2\Delta$ in $\run_5$.
In $\run_1$ and $\run_5$, respectively, the replicas $E$ and $D$ are Byzantine and act crashed from the second round.

Second, we construct $\run_3$, the execution in which the contradiction arises.
Both commands are submitted, the replicas in $A$ are Byzantine and act crashed after the first round, and the adversary splits the remaining replicas:
those in $B \cup D$ receive $x$ at time $\Delta$ and $y$ at time $2\Delta$, those in $C \cup E$ the converse.
During the first two rounds the former therefore behave exactly as they do in $\run_1$, and the latter exactly as they do in $\run_5$.
Since both clients are correct and $|A| \le f$, liveness forces every correct replica to eventually execute both commands;
let $R$ and $S$ be the rounds at which $b \in B$ and $c \in C$ execute their first command.

Third, we show that $b$ executes $x$ before $y$ in $\run_3$ while $c$ executes $y$ before $x$, contradicting safety.
For $b$, we build an execution $\run_2$ in which the replicas in $C$ are Byzantine and equivocate:
toward $A$ they behave as they do in $\run_1$, and toward the other replicas as they do in $\run_3$.
Replicas in $A$ thus cannot distinguish $\run_2$ from $\run_1$ through time $2\Delta$ and executes $x$ by that time, while external validity prevents it from having executed $y$; liveness and safety then force $b$ to execute $x$ before $y$ in $\run_2$.
Since $b$ cannot distinguish $\run_2$ from $\run_3$ through round $R$, the same holds in $\run_3$.
The execution $\run_4$ treats $c$ symmetrically, with the replicas in $B$ Byzantine and $\run_5$ in place of $\run_1$.

\par \textbf{Preliminaries.}
For some execution, we say that events that happen at the replicas during the time interval $[0, \Delta)$ form {\em the first round}, events that happen during the time interval $[\Delta, 2\Delta)$ {\em the second round}, and so on.
Given a set of replicas $R$ and an execution $\run$, $\steps{K}{R}{\run}$ denotes the steps taken by $R$ in the $K$-th round of $\run$.
Steps taken by Byzantine replicas are in \byz{red}.
An execution is \emph{timely}~\cite{two-step,revisitingBFTConsensus} when the faulty replicas behave correctly until time $\Delta$, after which they act crashed, the messages sent in a round are delivered precisely at the beginning of the next round ($GST=0$), and all local computations are instantaneous.
In a timely execution, because the protocol is deterministic, the steps of replicas $R$ during the first round are always the same;
they are written $\steps{1}{R}{}$.
%

\begin{theorem}[\leaderlesscaps{} BFT lower bound]
\label{thm:leaderless-lower-bound}
Any deterministic \leaderless BFT SMR protocol requires $n \ge 5f+1$ replicas.
\end{theorem}

\begin{proof}
Assume, for contradiction, that there exists a deterministic \leaderless BFT SMR protocol satisfying the stated properties with $n \le 5f$ replicas.
Let $\E$ be its set of executions.
Fix two valid conflicting commands $x$ and $y$. Since $3f+1 \leq n \le 5f$, partition the replicas into disjoint sets $A,\;B,\;C,\;D,\;E$ such that $|A|, |B|, |C|, |D|, |E| \le f$.
Fix replicas $a \in A$, $b \in B$, and $c \in C$.
We construct five executions in $\E$, denoted $\run_1,\run_2,\ldots,\run_5$.
Executions $\run_1$ and $\run_5$ are timely.

Since the protocol is deterministic, identical timed local histories
imply identical behavior.



\par \textbf{World $1$ ($\run_1$).}
This execution is timely.
Replicas in $E$ are Byzantine, and behave correctly until time $\Delta$, after which they act crashed.
A correct client $cl_x$ submits the valid command $x$ at time $0$ and no other command is submitted.
Every replica in $A \cup B \cup C \cup D$ receives $x$ at exactly time $\Delta$ (since the execution is timely).
Messages among all replicas in the first round are all received at the beginning of the second round.
Since $|E| \le f$, this is a post-GST execution with at most $f$ Byzantine replicas.
Therefore, by optimal execution latency, every correct replica in $A \cup B \cup C \cup D$ executes $x$ by time $2\Delta$.
In particular, replica $a$ executes $x$ by time $2\Delta$:
$$
\submit{x} \cdot \steps{1}{A \cup B \cup C \cup D}{} \cdot \byz{\steps{1}{E}{}} \cdot \steps{2}{A \cup B \cup C \cup D}{\run_1} \cdot \execute{x}{a} \in \E
$$

\par \textbf{World $5$ ($\run_5$).}
This execution is timely and symmetric to $\run_1$.
Replicas in $D$ are Byzantine, and behave correctly until time $\Delta$, after which they act crashed.
A correct client $cl_y \neq cl_x$ submits the valid command $y$ at time $0$ and no other command is submitted.
Every replica in $A \cup B \cup C \cup E$ receives $y$ at exactly time $\Delta$.
Messages among all replicas in the first round are received at the beginning of the second round.
Since $|D| \le f$, this is a post-GST execution with at most $f$ Byzantine replicas.
Therefore, by optimal execution latency, every correct replica in $A \cup B \cup C \cup E$ executes $y$ by time $2\Delta$.
In particular, replica $a$ executes $y$ by time $2\Delta$:
$$
\submit{y} \cdot \steps{1}{A \cup B \cup C \cup E}{} \cdot \byz{\steps{1}{D}{}} \cdot \steps{2}{A \cup B \cup C \cup E}{\run_5} \cdot \execute{y}{a} \in \E
$$

\par \textbf{World $3$ ($\run_3$).}
Replicas in $A$ are Byzantine, and behave correctly until time $\Delta$, after which they act crashed.
Unless stated otherwise, messages are scheduled so that when they are sent in one round, they are delivered at the beginning of the next round.
All local computations are instantaneous.
Commands $x$ and $y$ are submitted by the same two correct clients as in $\run_1$ and $\run_5$.

The adversary delays messages from clients so that replicas in $B \cup D$ receive $x$ at time $\Delta$ and $y$ at time $2\Delta$.
Conversely, replicas in $C \cup E$ receive $y$ at time $\Delta$ and $x$ at time $2\Delta$.
From what precedes, these two sets of replicas take the exact same first steps in $\run_3$ as the ones they do during the first two rounds of $\run_1$ and $\run_5$, respectively:
$$
\run'_3 = \submit{x} \cdot \submit{y} \cdot \steps{1}{B \cup C \cup D \cup E}{} \cdot \byz{\steps{1}{A}{}} \cdot \steps{2}{B \cup D}{\run_1} \cdot \steps{2}{C \cup E}{\run_5}  \in \E
$$

The two clients that submitted commands $x$ and $y$ are correct and only replicas in $A$ are Byzantine, with $|A| \le f$.
Hence any continuation from $\run'_3$ must uphold the properties of \leaderless BFT SMR.
In particular, by liveness, all correct replicas eventually execute both $x$ and $y$.

Let $R$ (respectively $S$) be the first round at which replica $b$ (resp., $c$) executes a command, either $x$ or $y$.
Without loss of generality, suppose $R \leq S$:
$$
\run'_3 \cdots \steps{R}{B \cup C \cup D \cup E }{\run_3} \cdots \steps{S}{B \cup C \cup D \cup E }{\run_3} \in \E
$$
We will show that, in execution $\run_3$, replica $b$ must execute $x$ before $y$, while replica $c$ must execute $y$ before $x$.

\par \textbf{World $2$ ($\run_2$).}
In $\run_2$, replicas in $C$ are Byzantine. Let $\mathrm{GST} > \max\{R\Delta,2\Delta\}$.
Thus, until time $\max\{R\Delta,2\Delta\}$, the adversary may delay messages arbitrarily.
Both $x$ and $y$ are submitted by clients.
These two commands are received during the first round by $C$.
The adversary delays messages from the clients so that replicas in $A \cup B \cup D$ receive $x$ exactly at time $\Delta$, and replicas in $E$ receive $y$ exactly at time $\Delta$.
The adversary delays any message containing or referring to $y$ to replica $a$ until after time $2\Delta$.
In the first round, replicas in $A \cup B \cup D$ (respectively, $E$) take the exact same steps as they do during the first round of $\run_1$ (resp., $\run_5$).
$$
\run'_2 = \submit{x} \cdot \submit{y} \cdot \steps{1}{A \cup B \cup D}{} \cdot \steps{1}{E}{} \in \E
$$

The Byzantine replicas in $C$ equivocate as follows. Toward replicas in $A$, they behave exactly as the correct replicas in $C$ behave in $\run_1$ through time $2\Delta$; in particular, they act as if they received $x$ at exactly time $\Delta$. Toward replicas outside $A$, they behave exactly as the correct replicas in $C$ behave in $\run_3$ through time $R\Delta$; in particular, they act as if they received $y$ at exactly time $\Delta$.

Messages from replicas in $A$ to replicas outside $A$ that are sent in later rounds than the first one are delayed until after GST.
The remaining messages are scheduled so that replica $a$'s timed local history through time $2\Delta$ is identical as in $\run_1$, and replica $b$'s timed local history through time $R\Delta$ is identical as in $\run_3$.

The first indistinguishability holds because, through time $2\Delta$, replica $a$ receives first round messages from all replicas and $x$ at exactly time $\Delta$, receives the same second round messages from replicas in $A \cup B \cup C \cup D$ as in $\run_1$, receives no additional messages from $E$, and receives no message carrying or referring to $y$.

Therefore, since $a$ executes $x$ by time $2\Delta$ in $\run_1$, determinism implies that $a$ also executes $x$ by time $2\Delta$ in $\run_2$. Since $a$ has identical timed histories in $\run_1$ and $\run_2$ through $2\Delta$ and does not execute $y$ in $\run_1$ (no client submitted $y$ in $\run_1$, so by external validity $a$ cannot execute $y$ in $\run_1$), $a$ could not have executed $y$ by $2\Delta$ in $\run_2$.
In other words,
$$
\run''_2 = \run'_2 \cdot \byz{\steps{1}{C}{} \cdot \steps{2}{C}{\run_1}} \cdot \steps{2}{A}{\run_1} \cdot \steps{2}{B \cup D}{\run_1} \cdot \execute{x}{a}  \in \E
$$

The second indistinguishability holds because, in $\run_3$, all replicas receive the same first round messages from all replicas, replicas in $A$ are Byzantine and silent after the first round, while in $\run_2$ all messages after the first round from replicas in $A$ to replicas outside $A$ are delayed until after time $R\Delta$.
Moreover, replicas in $B \cup D$ receive commands $x$ and $y$ at respectively times $\Delta$ and $2\Delta$ in both executions, replicas in $E$ receive $y$ at exactly time $\Delta$ in both executions, and Byzantine replicas in $C$ send to replicas outside $A$ exactly the messages that correct replicas in $C$ send in $\run_3$ through time $R\Delta$.
$$
\run_2 = \run''_2 \cdot \byz{\steps{2}{C}{\run_5}} \cdot \steps{2}{E }{\run_5} \cdots \steps{R}{B \cup C \cup D \cup E }{\run_3} \in \E
$$

Since replica $a$ executes $x$ before $y$, replica $b$ must execute $x$ before $y$.
Thus, replica $b$ applies $x$ in round $R$: $\execute{x}{b} \in \steps{R}{b }{\run_3}$. 





\par \textbf{World $4$ ($\run_4$).}
The construction is symmetric to $\run_2$.
The Byzantine replicas in $B$ equivocate as follows.
Toward replicas in $A$, they behave exactly as the correct replicas in $B$ behave in $\run_5$ through time $2\Delta$; in particular, they act as if they received $y$ at exactly time $\Delta$.
Toward replicas outside $A$, they behave exactly as the correct replicas in $B$ behave in $\run_3$ through time $S\Delta$; in particular, they act as if they received $x$ at exactly time $\Delta$.
Below, we list the equations which lead to the fact that replica $c$ executes $y$ first in $\run_3$:
$$
\begin{array}{l}
  \run'_4 = \submit{x} \cdot \submit{y} \cdot \steps{1}{A \cup C \cup E}{} \cdot \steps{1}{D}{} \in \E \\[.5em]
  
  \run''_4 = \run'_4 \cdot \byz{\steps{1}{B}{} \cdot \steps{2}{B}{\run_5}} \cdot \steps{2}{A}{\run_5} \cdot \steps{2}{C \cup E}{\run_5} \cdot \execute{y}{a} \in \E \\[.5em]
  
  \run_4 = \run''_4 \cdot \byz{\steps{2}{B}{\run_1}} \cdot \steps{2}{D}{\run_1} \cdots \steps{S}{B \cup C \cup D \cup E }{\run_3} \in \E
\end{array}
$$

\par \textbf{Conclusion.}
We have established that in $\run_3$ replica $b$ executes $x$ before $y$, while replica $c$ executes $y$ before $x$.
Since both replicas are correct, $\run_3$ violates the Safety property of BFT SMR;
contradiction.

\end{proof}

\section{$n\geq 7f+1$ Protocol}

In this section we present a simpler, but suboptimal, generalized leaderless SMR protocol for the case $n\geq 7f+1$. The benefit of this protocol is that the larger replication factor makes the validation phase unnecessary. Thus, the $n\geq 7f+1$ protocol exposes a useful trade-off: it requires more replicas, but has a simpler and more efficient recovery path. In settings where recovery overhead is a throughput bottleneck, this simpler recovery procedure may lead to higher throughput per machine despite the larger replication factor~\cite{suyash}.

The difficulty in the $n\geq 5f+1$ protocol is that recovery may only see $n-3f\geq 2f+1$ matching \textsc{FastVote} messages for the committed value. When $n\geq 5f+1$, two quorums of size $n-3f$ may not intersect in a correct replica. Thus, two conflicting commands could each have enough \textsc{FastVote} messages to be proposed during recovery, yet the replicas that voted for one command might contain no correct replica that also voted for the other. Thus, the $n\geq 5f+1$ protocol adds an additional validation phase to ensure that conflicting commands cannot miss each other during recovery.

For $n\geq 7f+1$, this problem disappears. Any two sets of $n-3f$ replicas intersect in at least one correct replica. Therefore, if two conflicting commands each have $n-3f$ \textsc{FastVote} messages, then some correct replica fast-voted for both commands. A correct replica includes every known conflicting command in the dependency set of any new command for which it sends a \textsc{FastVote} message. Hence, whichver command the replica voted for second must incude the first as a dependency. We can leverage this guarantee to simplify the recovery protocol and fully remove the validation phase.

The Status phase of the protocol thus changes in the following way. If the coordinator observes $n-3f$ matching prior \textsc{Vote} messages, it reproposes that value. Similarly, if it observes $n-3f$ matching \textsc{FastVote} messages, it can directly repropose that value. Otherwise, if the the fast votes disagree on the dependency set; the coordinator the command with the union of all reported dependencies. This is correct because the $n-3f$ quorums intersect with each other. All other phases of the protocol remain unchanged.

Algorithm~\ref{alg:replica-sevenf-plusone} highlights the differences from the $n\geq 5f+1$ protocol.





\begin{algorithm}[htbp]
\caption{Modified Algorithm~\ref{alg:replica-part-1} for $n\geq 7f+1$}
\label{alg:replica-sevenf-plusone}
\begin{algorithmic}[1]

\Statex \diffnote{All local state and handlers are unchanged from
Algorithm~\ref{alg:replica-part-1}, except that the status-phase handler is
replaced by the following.}

\medskip
\Event{receiving $S\gets n-f$ valid \msg{Status}{id, v, \cdot,\cdot} and $i=\Call{Coordinator}{id,v}$}
    \If{there exists $n-3f$ $votes[id]$ with $val$ from $S$}
        \State \textbf{send} \msg{Propose}{id,v,val,S,\bot,\bot} to all
        \Comment{Possibly committed in an earlier view}

    \ElsIf{there exists $n-3f$ $fastvotes[id]$ with $val=(x,deps)$ from $S$}
        \State \diffdel{\textbf{send} \msg{Validate}{id, v, (x,deps), S} to all}
        \State \diffadd{\textbf{send} \msg{Propose}{id,v,val,S,\bot,\bot} to all}
        \Comment{\diffadd{Validation phase is unnecessary}}
        
    \Else
        \State $x\gets$ unique non-equivocating command in the collected \textsc{FastVote} messages
        \State $D\gets$ union of all dependencies from $fv[id]$
        \State \textbf{send} \msg{Propose}{id,v,(x,D),S,\bot,\bot} to all
    \EndIf
\EndEvent

\Statex
\Statex \diffnote{The validation-phase handlers from Algorithm~\ref{alg:replica-part-1}
are removed. Furthemore, the view timer is set to $4\Delta$ instead of $6\Delta$.}

\end{algorithmic}
\end{algorithm}

\section{Related Work}
\label{sec:related-work}

We compare leaderless Byzantine state-machine replication (SMR) with several related lines of work. 

\subsection{CFT Protocols}

\par \textbf{Crash-fault-tolerant generalized SMR.}
The EPaxos family of protocols satisfies a leaderless generalized SMR style of
execution in the crash-fault-tolerant setting. EPaxos~\cite{epaxos} and its
successors, including Atlas~\cite{atlas}, Caesar~\cite{caesar},
EPaxos-TOQ~\cite{epaxos-revisited}, Accord~\cite{accord},
EPaxos*~\cite{epaxos-star}, and BPaxos~\cite{bpaxos}, allow replicas to
propose commands independently and order commands by tracking dependencies
between conflicting operations. These protocols are the closest conceptual
predecessors to the abstraction studied in this paper. However, they rely
heavily on the crash-fault model. Byzantine faults fundamentally complicate the dependency agreement
problem, and crash-fault-tolerant techniques do not directly extend to the BFT
setting. This paper studies the additional cost of obtaining EPaxos-style
leaderless SMR in the presence of Byzantine behavior.

\subsection{BFT Protocols}
None of the work mentioned below achieve the optimal two-message-delay execution latency for every conflict-free command during synchronous periods while tolerating up to $f$ failures.

\par \textbf{Leader-based BFT.}
Classical BFT SMR protocols such as PBFT~\cite{pbft} and HotStuff~\cite{hotstuff} order commands through a leader. This design achieves
the optimal Byzantine resilience bound of $n = 3f+1$, but commands must pass
through a leader-controlled agreement path before they can be executed. In PBFT-style protocols, a client first sends its request to the
leader, after which the replicas run several communication phases to order and
commit the command. As a
result, even conflict-free commands take multiple message delays
to execute. 

PBFT~\cite{pbft} and other leader-based protocols like HotStuff-1~\cite{hotstuff1}, and PoE~\cite{poe} also support speculative execution. These protocols can execute
early by allowing the system to roll back, repair, or later validate tentative
execution results. Our execution latency property concerns final execution:
when a correct replica executes a command, that execution is stable and is not
later rolled back.

\par \textbf{Three message-delay fast path.}
A large body of work reduces the optimistic latency of leader-based BFT
protocols. Examples include FaB Paxos~\cite{fab},
Zyzzyva~\cite{zyzzyva}, SBFT~\cite{sbft},
Kudzu~\cite{kudzu}, and Alpenglow~\cite{alpenglow}. These protocols have a replication factor of $n\geq 3f+2p+1$, where $p$ is the number of faulty replicas that can be tolerated on the fast path. They provide a fast path in which a proposal can be committed in three total message delays including the client's initial request. This is faster than traditional PBFT-style protocols, but the fast path remains leader-based. 

\par \textbf{Two message delays + epsilon fast path.}
Flutter/Blink~\cite{flutter} and Aspen~\cite{aspen} execute conflict-free commands without being sequenced through a leader, but they rely on tightly synchronized clocks. Their conflict-free latency under synchrony is $2\Delta + \epsilon$, where $\epsilon$ accounts for clock discrepancies. In contrast, our protocol achieves latency $2\Delta$ without requiring tightly synchronized clocks.

\par \textbf{Optimistic two message delay fast path.} Other protocols achieve low latency only under more restrictive failure
assumptions. Wintermute~\cite{leaderless-thesis} and Basil~\cite{basil} have fast paths that can execute
conflict-free commands in two message delays, but that fast path does not tolerate
failures. Q/U~\cite{qu} and Aliph~\cite{aliph} similarly 
not tolerate failures on its fast path. These protocols demonstrate that
leaderless conflict-free execution is possible in favorable executions, but
they do not provide the same failure-tolerant optimal execution latency guarantee required by
our definition.

\section{Conclusion}
Leaderless BFT SMR has historically been less well understood than its crash-fault-tolerant counterpart. This paper narrows that gap by establishing tight bounds on the replication factor required for leaderless BFT SMR. We also illustrate a trade-off between replication cost and recovery efficiency by presenting a simpler protocol with a larger replication factor.

\bibliographystyle{plain}
\bibliography{references}

@inproceedings{raft,
author = {Ongaro, Diego and Ousterhout, John},
title = {In search of an understandable consensus algorithm},
year = {2014},
isbn = {9781931971102},
publisher = {USENIX Association},
address = {USA},
booktitle = {Proceedings of the 2014 USENIX Conference on USENIX Annual Technical Conference},
pages = {305–320},
numpages = {16},
location = {Philadelphia, PA},
series = {USENIX ATC'14}
}

@article{multi-paxos,
author = {Van Renesse, Robbert and Altinbuken, Deniz},
title = {Paxos Made Moderately Complex},
year = {2015},
issue_date = {April 2015},
publisher = {Association for Computing Machinery},
address = {New York, NY, USA},
volume = {47},
number = {3},
issn = {0360-0300},
url = {https://doi.org/10.1145/2673577},
doi = {10.1145/2673577},
journal = {ACM Comput. Surv.},
month = feb,
articleno = {42},
numpages = {36}
}

@article{linearizability,
author = {Herlihy, Maurice P. and Wing, Jeannette M.},
title = {Linearizability: a correctness condition for concurrent objects},
year = {1990},
issue_date = {July 1990},
publisher = {Association for Computing Machinery},
address = {New York, NY, USA},
volume = {12},
number = {3},
issn = {0164-0925},
url = {https://doi.org/10.1145/78969.78972},
doi = {10.1145/78969.78972},
journal = {ACM Trans. Program. Lang. Syst.},
month = jul,
pages = {463–492},
numpages = {30}
}

@inproceedings{epaxos,
author = {Moraru, Iulian and Andersen, David G. and Kaminsky, Michael},
title = {There is more consensus in Egalitarian parliaments},
year = {2013},
isbn = {9781450323888},
publisher = {Association for Computing Machinery},
address = {New York, NY, USA},
url = {https://doi.org/10.1145/2517349.2517350},
doi = {10.1145/2517349.2517350},
booktitle = {Proceedings of the Twenty-Fourth ACM Symposium on Operating Systems Principles},
pages = {358–372},
numpages = {15},
location = {Farminton, Pennsylvania},
series = {SOSP '13}
}

@inproceedings {epaxos-revisited,
author = {Sarah Tollman and Seo Jin Park and John Ousterhout},
title = {{EPaxos} Revisited},
booktitle = {18th USENIX Symposium on Networked Systems Design and Implementation (NSDI 21)},
year = {2021},
isbn = {978-1-939133-21-2},
pages = {613--632},
url = {https://www.usenix.org/conference/nsdi21/presentation/tollman},
publisher = {USENIX Association},
month = apr
}

@article{smr,
  author     = {Schneider, Fred B.},
  title      = {Implementing Fault-Tolerant Services Using the State Machine Approach: A Tutorial},
  journal    = {ACM Computing Surveys},
  volume     = {22},
  number     = {4},
  pages      = {299--319},
  month      = dec,
  year       = {1990},
  publisher  = {ACM New York, NY, USA},
  doi        = {10.1145/98163.98167},
  url        = {https://dl.acm.org/doi/10.1145/98163.98167}
}

@inproceedings{atlas,
author = {Enes, Vitor and Baquero, Carlos and Rezende, Tuanir Fran\c{c}a and Gotsman, Alexey and Perrin, Matthieu and Sutra, Pierre},
title = {State-machine replication for planet-scale systems},
year = {2020},
isbn = {9781450368827},
publisher = {Association for Computing Machinery},
address = {New York, NY, USA},
url = {https://doi.org/10.1145/3342195.3387543},
doi = {10.1145/3342195.3387543},
booktitle = {Proceedings of the Fifteenth European Conference on Computer Systems},
articleno = {24},
numpages = {15},
location = {Heraklion, Greece},
series = {EuroSys '20}
}

@INPROCEEDINGS{caesar,
  author={Arun, Balaji and Peluso, Sebastiano and Palmieri, Roberto and Losa, Giuliano and Ravindran, Binoy},
  booktitle={2017 47th Annual IEEE/IFIP International Conference on Dependable Systems and Networks (DSN)}, 
  title={Speeding up Consensus by Chasing Fast Decisions}, 
  year={2017},
  volume={},
  number={},
  pages={49-60},
  doi={10.1109/DSN.2017.35}}

@article{bpaxos,
  author       = {Michael J. Whittaker and
                  Neil Giridharan and
                  Adriana Szekeres and
                  Joseph M. Hellerstein and
                  Ion Stoica},
  title        = {SoK: {A} Generalized Multi-Leader State Machine Replication Tutorial},
  journal      = {J. Syst. Res.},
  volume       = {1},
  number       = {1},
  year         = {2021},
  url          = {https://doi.org/10.5070/sr31154817},
  doi          = {10.5070/SR31154817},
  bibsource    = {dblp computer science bibliography, https://dblp.org}
}

@misc{accord,
  author       = {Benedict Elliott Smith and Tony Zhang and Blake Eggleston and Scott Andreas},
  title        = {CEP-15: Fast General Purpose Transactions},
  year         = {2021},
  howpublished = {Apache Cassandra Enhancement Proposal (CEP-15) Whitepaper},
  url          = {https://cwiki.apache.org/confluence/download/attachments/188744725/Accord.pdf},
}

@software{cassandra,
  author       = {{Apache Software Foundation}},
  title        = {Apache Cassandra},
  year         = {2026},
  url          = {https://cassandra.apache.org/},
  organization = {Apache Software Foundation},
  note         = {Version 5.0, accessed 2026-06-01}
}

@InProceedings{leaderless-smr,
  author ={Tuanir Fran{\c{c}}a Rezende and Pierre Sutra},
  title ={{Leaderless State-Machine Replication: Specification, Properties, Limits}},
  booktitle ={34th International Symposium on Distributed Computing (DISC 2020)},
  pages ={24:1--24:17},
  series ={Leibniz International Proceedings in Informatics (LIPIcs)},
  ISBN ={978-3-95977-168-9},
  ISSN ={1868-8969},
  year ={2020},
  volume ={179},
  editor ={Hagit Attiya},
  publisher ={Schloss Dagstuhl--Leibniz-Zentrum f{\"u}r Informatik},
  address ={Dagstuhl, Germany},
  URL ={https://drops.dagstuhl.de/opus/volltexte/2020/13102},
  URN ={urn:nbn:de:0030-drops-131024},
  doi ={10.4230/LIPIcs.DISC.2020.24}
}

@techreport{generalized-paxos,
author = {Lamport, Leslie},
title = {Generalized Consensus and Paxos},
year = {2005},
month = {March},
url = {https://www.microsoft.com/en-us/research/publication/generalized-consensus-and-paxos/},
pages = {60},
number = {MSR-TR-2005-33},
}

@misc{byzantine-generalized-paxos,
      title={Generalized Paxos Made Byzantine (and Less Complex)}, 
      author={Miguel Pires and Srivatsan Ravi and Rodrigo Rodrigues},
      year={2017},
      eprint={1708.07575},
      archivePrefix={arXiv},
      primaryClass={cs.DC},
      url={https://arxiv.org/abs/1708.07575}, 
}

@InProceedings{epaxos-star,
  author =	{Ryabinin, Fedor and Gotsman, Alexey and Sutra, Pierre},
  title =	{{Making Democracy Work: Fixing and Simplifying Egalitarian Paxos}},
  booktitle =	{29th International Conference on Principles of Distributed Systems (OPODIS 2025)},
  pages =	{22:1--22:19},
  series =	{Leibniz International Proceedings in Informatics (LIPIcs)},
  ISBN =	{978-3-95977-409-3},
  ISSN =	{1868-8969},
  year =	{2026},
  volume =	{361},
  editor =	{Arusoaie, Andrei and Onica, Emanuel and Spear, Michael and Tucci-Piergiovanni, Sara},
  publisher =	{Schloss Dagstuhl -- Leibniz-Zentrum f{\"u}r Informatik},
  address =	{Dagstuhl, Germany},
  URL =		{https://drops.dagstuhl.de/entities/document/10.4230/LIPIcs.OPODIS.2025.22},
  URN =		{urn:nbn:de:0030-drops-251955},
  doi =		{10.4230/LIPIcs.OPODIS.2025.22}
}

@misc{bipartisan-paxos,
      title={Bipartisan Paxos: A Modular State Machine Replication Protocol}, 
      author={Michael Whittaker and Neil Giridharan and Adriana Szekeres and Joseph M. Hellerstein and Ion Stoica},
      year={2020},
      eprint={2003.00331},
      archivePrefix={arXiv},
      primaryClass={cs.DC},
      url={https://arxiv.org/abs/2003.00331}, 
}

@inproceedings{pbft,
author = {Castro, Miguel and Liskov, Barbara},
title = {Practical Byzantine fault tolerance},
year = {1999},
isbn = {1880446391},
publisher = {USENIX Association},
address = {USA},
booktitle = {Proceedings of the Third Symposium on Operating Systems Design and Implementation},
pages = {173–186},
numpages = {14},
location = {New Orleans, Louisiana, USA},
series = {OSDI '99}
}

@inproceedings{hotstuff,
author = {Yin, Maofan and Malkhi, Dahlia and Reiter, Michael K. and Gueta, Guy Golan and Abraham, Ittai},
title = {HotStuff: BFT Consensus with Linearity and Responsiveness},
year = {2019},
isbn = {9781450362177},
publisher = {Association for Computing Machinery},
address = {New York, NY, USA},
url = {https://doi.org/10.1145/3293611.3331591},
doi = {10.1145/3293611.3331591},
booktitle = {Proceedings of the 2019 ACM Symposium on Principles of Distributed Computing},
pages = {347–356},
numpages = {10},
location = {Toronto ON, Canada},
series = {PODC '19}
}

@ARTICLE{fab,
author={Alvisi, Lorenzo and Martin, Jean-Philippe},
journal={ IEEE Transactions on Dependable and Secure Computing },
title={{ Fast Byzantine Consensus }},
year={2006},
volume={3},
number={03},
ISSN={1941-0018},
pages={202-215},
doi={10.1109/TDSC.2006.35},
url = {https://doi.ieeecomputersociety.org/10.1109/TDSC.2006.35},
publisher={IEEE Computer Society},
address={Los Alamitos, CA, USA},
month=jul}

@article{zyzzyva,
author = {Kotla, Ramakrishna and Alvisi, Lorenzo and Dahlin, Mike and Clement, Allen and Wong, Edmund},
title = {Zyzzyva: Speculative Byzantine fault tolerance},
year = {2010},
issue_date = {December 2009},
publisher = {Association for Computing Machinery},
address = {New York, NY, USA},
volume = {27},
number = {4},
issn = {0734-2071},
url = {https://doi.org/10.1145/1658357.1658358},
doi = {10.1145/1658357.1658358},
journal = {ACM Trans. Comput. Syst.},
month = jan,
articleno = {7},
numpages = {39}
}

@INPROCEEDINGS{sbft,
  author={Golan Gueta, Guy and Abraham, Ittai and Grossman, Shelly and Malkhi, Dahlia and Pinkas, Benny and Reiter, Michael and Seredinschi, Dragos-Adrian and Tamir, Orr and Tomescu, Alin},
  booktitle={2019 49th Annual IEEE/IFIP International Conference on Dependable Systems and Networks (DSN)}, 
  title={SBFT: A Scalable and Decentralized Trust Infrastructure}, 
  year={2019},
  volume={},
  number={},
  pages={568-580},
  doi={10.1109/DSN.2019.00063}}

@InProceedings{kudzu,
  author =	{Shoup, Victor and Sliwinski, Jakub and Vonlanthen, Yann},
  title =	{{Kudzu: Fast and Simple High-Throughput BFT}},
  booktitle =	{39th International Symposium on Distributed Computing (DISC 2025)},
  pages =	{42:1--42:19},
  series =	{Leibniz International Proceedings in Informatics (LIPIcs)},
  ISBN =	{978-3-95977-402-4},
  ISSN =	{1868-8969},
  year =	{2025},
  volume =	{356},
  editor =	{Kowalski, Dariusz R.},
  publisher =	{Schloss Dagstuhl -- Leibniz-Zentrum f{\"u}r Informatik},
  address =	{Dagstuhl, Germany},
  URL =		{https://drops.dagstuhl.de/entities/document/10.4230/LIPIcs.DISC.2025.42},
  URN =		{urn:nbn:de:0030-drops-248597},
  doi =		{10.4230/LIPIcs.DISC.2025.42}
}

@inproceedings{bosco,
author = {Song, Yee Jiun and Renesse, Robbert},
title = {Bosco: One-Step Byzantine Asynchronous Consensus},
year = {2008},
isbn = {9783540877783},
publisher = {Springer-Verlag},
address = {Berlin, Heidelberg},
url = {https://doi.org/10.1007/978-3-540-87779-0_30},
doi = {10.1007/978-3-540-87779-0_30},
booktitle = {Proceedings of the 22nd International Symposium on Distributed Computing},
pages = {438–450},
numpages = {13},
location = {Arcachon, France},
series = {DISC '08}
}

@article{flutter,
  title={Fast Leaderless Byzantine Total Order Broadcast},
  author={Monti, Matteo and Camaioni, Martina and Roman, Pierre-Louis},
  journal={arXiv preprint arXiv:2412.14061},
  year={2024},
  url={https://arxiv.org/abs/2412.14061}
}

@misc{aspen,
      title={Revisiting Speculative Leaderless Protocols for Low-Latency BFT Replication}, 
      author={Daniel Qian and Xiyu Hao and Jinkun Geng and Yuncheng Yao and Aurojit Panda and Jinyang Li and Anirudh Sivaraman},
      year={2026},
      eprint={2601.03390},
      archivePrefix={arXiv},
      primaryClass={cs.DC},
      url={https://arxiv.org/abs/2601.03390}, 
}

@phdthesis{leaderless-thesis,
  TITLE = {{Leaderless state-machine replication : from fail-stop to Byzantine failures}},
  AUTHOR = {Franca Rezende, Tuanir},
  URL = {https://theses.hal.science/tel-03584254},
  NUMBER = {2021IPPAS016},
  SCHOOL = {{Institut Polytechnique de Paris}},
  YEAR = {2021},
  MONTH = Dec,
  TYPE = {Theses},
  HAL_ID = {tel-03584254},
  HAL_VERSION = {v1},
}

@inproceedings{basil,
author = {Suri-Payer, Florian and Burke, Matthew and Wang, Zheng and Zhang, Yunhao and Alvisi, Lorenzo and Crooks, Natacha},
title = {Basil: Breaking up BFT with ACID (transactions)},
year = {2021},
isbn = {9781450387095},
publisher = {Association for Computing Machinery},
address = {New York, NY, USA},
url = {https://doi.org/10.1145/3477132.3483552},
doi = {10.1145/3477132.3483552},
booktitle = {Proceedings of the ACM SIGOPS 28th Symposium on Operating Systems Principles},
pages = {1–17},
numpages = {17},
location = {Virtual Event, Germany},
series = {SOSP '21}
}

@misc{mir-bft,
      title={Mir-BFT: High-Throughput Robust BFT for Decentralized Networks}, 
      author={Chrysoula Stathakopoulou and Tudor David and Matej Pavlovic and Marko Vukolić},
      year={2021},
      eprint={1906.05552},
      archivePrefix={arXiv},
      primaryClass={cs.DC},
      url={https://arxiv.org/abs/1906.05552}, 
}

@inproceedings{iss,
author = {Stathakopoulou, Chrysoula and Pavlovic, Matej and Vukoli\'{c}, Marko},
title = {State machine replication scalability made simple},
year = {2022},
isbn = {9781450391627},
publisher = {Association for Computing Machinery},
address = {New York, NY, USA},
url = {https://doi.org/10.1145/3492321.3519579},
doi = {10.1145/3492321.3519579},
booktitle = {Proceedings of the Seventeenth European Conference on Computer Systems},
pages = {17–33},
numpages = {17},
location = {Rennes, France},
series = {EuroSys '22}
}

@INPROCEEDINGS{bft-mencius,
  author={Milosevic, Zarko and Biely, Martin and Schiper, André},
  booktitle={2013 IEEE 32nd International Symposium on Reliable Distributed Systems}, 
  title={Bounded Delay in Byzantine-Tolerant State Machine Replication}, 
  year={2013},
  volume={},
  number={},
  pages={61-70},
  doi={10.1109/SRDS.2013.15}}

@INPROCEEDINGS{rcc,
  author={Gupta, Suyash and Hellings, Jelle and Sadoghi, Mohammad},
  booktitle={2021 IEEE 37th International Conference on Data Engineering (ICDE)}, 
  title={RCC: Resilient Concurrent Consensus for High-Throughput Secure Transaction Processing}, 
  year={2021},
  volume={},
  number={},
  pages={1392-1403},
  doi={10.1109/ICDE51399.2021.00124}}

@inproceedings{bullshark,
author = {Spiegelman, Alexander and Giridharan, Neil and Sonnino, Alberto and Kokoris-Kogias, Lefteris},
title = {Bullshark: DAG BFT Protocols Made Practical},
year = {2022},
isbn = {9781450394505},
publisher = {Association for Computing Machinery},
address = {New York, NY, USA},
url = {https://doi.org/10.1145/3548606.3559361},
doi = {10.1145/3548606.3559361},
booktitle = {Proceedings of the 2022 ACM SIGSAC Conference on Computer and Communications Security},
pages = {2705–2718},
numpages = {14},
location = {Los Angeles, CA, USA},
series = {CCS '22}
}

@inproceedings{autobahn,
author = {Giridharan, Neil and Suri-Payer, Florian and Abraham, Ittai and Alvisi, Lorenzo and Crooks, Natacha},
title = {Autobahn: Seamless high speed BFT},
year = {2024},
isbn = {9798400712517},
publisher = {Association for Computing Machinery},
address = {New York, NY, USA},
url = {https://doi.org/10.1145/3694715.3695942},
doi = {10.1145/3694715.3695942},
booktitle = {Proceedings of the ACM SIGOPS 30th Symposium on Operating Systems Principles},
pages = {1–23},
numpages = {23},
location = {Austin, TX, USA},
series = {SOSP '24}
}

@inproceedings{narwhal-tusk,
author = {Danezis, George and Kokoris-Kogias, Lefteris and Sonnino, Alberto and Spiegelman, Alexander},
title = {Narwhal and Tusk: a DAG-based mempool and efficient BFT consensus},
year = {2022},
isbn = {9781450391627},
publisher = {Association for Computing Machinery},
address = {New York, NY, USA},
url = {https://doi.org/10.1145/3492321.3519594},
doi = {10.1145/3492321.3519594},
booktitle = {Proceedings of the Seventeenth European Conference on Computer Systems},
pages = {34–50},
numpages = {17},
location = {Rennes, France},
series = {EuroSys '22}
}

@inproceedings{aliph,
author = {Guerraoui, Rachid and Kne\v{z}evi\'{c}, Nikola and Qu\'{e}ma, Vivien and Vukoli\'{c}, Marko},
title = {The next 700 BFT protocols},
year = {2010},
isbn = {9781605585772},
publisher = {Association for Computing Machinery},
address = {New York, NY, USA},
url = {https://doi.org/10.1145/1755913.1755950},
doi = {10.1145/1755913.1755950},
booktitle = {Proceedings of the 5th European Conference on Computer Systems},
pages = {363–376},
numpages = {14},
location = {Paris, France},
series = {EuroSys '10}
}

@misc{poe,
      title={Proof-of-Execution: Reaching Consensus through Fault-Tolerant Speculation}, 
      author={Suyash Gupta and Jelle Hellings and Sajjad Rahnama and Mohammad Sadoghi},
      year={2021},
      eprint={1911.00838},
      archivePrefix={arXiv},
      primaryClass={cs.DB},
      url={https://arxiv.org/abs/1911.00838}, 
}

@article{rachid-leaderless,
title = {Leaderless consensus},
journal = {Journal of Parallel and Distributed Computing},
volume = {176},
pages = {95-113},
year = {2023},
issn = {0743-7315},
doi = {https://doi.org/10.1016/j.jpdc.2023.01.009},
url = {https://www.sciencedirect.com/science/article/pii/S0743731523000151},
author = {Karolos Antoniadis and Julien Benhaim and Antoine Desjardins and Elias Poroma and Vincent Gramoli and Rachid Guerraoui and Gauthier Voron and Igor Zablotchi}
}

@misc{two-step,
      title={Revisiting Lower Bounds for Two-Step Consensus}, 
      author={Fedor Ryabinin and Alexey Gotsman and Pierre Sutra},
      year={2026},
      eprint={2505.03627},
      archivePrefix={arXiv},
      primaryClass={cs.DC},
      url={https://arxiv.org/abs/2505.03627}, 
}

@inproceedings{revisitingBFTConsensus,
author = {Kuznetsov, Petr and Tonkikh, Andrei and Zhang, Yan X},
title = {Revisiting Optimal Resilience of Fast Byzantine Consensus},
year = {2021},
isbn = {9781450385480},
publisher = {Association for Computing Machinery},
address = {New York, NY, USA},
url = {https://doi.org/10.1145/3465084.3467924},
doi = {10.1145/3465084.3467924},
booktitle = {Proceedings of the 2021 ACM Symposium on Principles of Distributed Computing},
pages = {343–353},
numpages = {11},
location = {Virtual Event, Italy},
series = {PODC'21}
}

@misc{ezbft,
      title={ezBFT: Decentralizing Byzantine Fault-Tolerant State Machine Replication}, 
      author={Balaji Arun and Sebastiano Peluso and Binoy Ravindran},
      year={2019},
      eprint={1904.06023},
      archivePrefix={arXiv},
      primaryClass={cs.DC},
      url={https://arxiv.org/abs/1904.06023}, 
}

@misc{ezbft-revisited,
      title={Revisiting EZBFT: A Decentralized Byzantine Fault Tolerant Protocol with Speculation}, 
      author={Nibesh Shrestha and Mohan Kumar},
      year={2019},
      eprint={1909.03990},
      archivePrefix={arXiv},
      primaryClass={cs.DC},
      url={https://arxiv.org/abs/1909.03990}, 
}

@INPROCEEDINGS{ebft,
  author={Eischer, Michael and Distler, Tobias},
  booktitle={2021 IEEE 26th Pacific Rim International Symposium on Dependable Computing (PRDC)}, 
  title={Egalitarian Byzantine Fault Tolerance}, 
  year={2021},
  volume={},
  number={},
  pages={1-10},
  doi={10.1109/PRDC53464.2021.00019}}

@inproceedings{suyash,
author = {Gupta, Suyash and Rahnama, Sajjad and Pandey, Shubham and Crooks, Natacha and Sadoghi, Mohammad},
title = {Dissecting BFT Consensus: In Trusted Components we Trust!},
year = {2023},
isbn = {9781450394871},
publisher = {Association for Computing Machinery},
address = {New York, NY, USA},
url = {https://doi.org/10.1145/3552326.3587455},
doi = {10.1145/3552326.3587455},
booktitle = {Proceedings of the Eighteenth European Conference on Computer Systems},
pages = {521–539},
numpages = {19},
location = {Rome, Italy},
series = {EuroSys '23}
}

@misc{hotstuff1,
      title={HotStuff-1: Linear Consensus with One-Phase Speculation}, 
      author={Dakai Kang and Suyash Gupta and Dahlia Malkhi and Mohammad Sadoghi},
      year={2025},
      eprint={2408.04728},
      archivePrefix={arXiv},
      primaryClass={cs.DB},
      url={https://arxiv.org/abs/2408.04728}, 
}

@inproceedings{honeybadger,
author = {Miller, Andrew and Xia, Yu and Croman, Kyle and Shi, Elaine and Song, Dawn},
title = {The Honey Badger of BFT Protocols},
year = {2016},
isbn = {9781450341394},
publisher = {Association for Computing Machinery},
address = {New York, NY, USA},
url = {https://doi.org/10.1145/2976749.2978399},
doi = {10.1145/2976749.2978399},
booktitle = {Proceedings of the 2016 ACM SIGSAC Conference on Computer and Communications Security},
pages = {31–42},
numpages = {12},
location = {Vienna, Austria},
series = {CCS '16}
}

@inproceedings{vaba,
author = {Abraham, Ittai and Malkhi, Dahlia and Spiegelman, Alexander},
title = {Asymptotically Optimal Validated Asynchronous Byzantine Agreement},
year = {2019},
isbn = {9781450362177},
publisher = {Association for Computing Machinery},
address = {New York, NY, USA},
url = {https://doi.org/10.1145/3293611.3331612},
doi = {10.1145/3293611.3331612},
booktitle = {Proceedings of the 2019 ACM Symposium on Principles of Distributed Computing},
pages = {337–346},
numpages = {10},
location = {Toronto ON, Canada},
series = {PODC '19}
}

@misc{byblos,
      title={Clairvoyant State Machine Replication}, 
      author={Rida Bazzi and Maurice Herlihy},
      year={2019},
      eprint={1905.11607},
      archivePrefix={arXiv},
      primaryClass={cs.DC},
      url={https://arxiv.org/abs/1905.11607}, 
}

@INPROCEEDINGS{dbft,
  author={Crain, Tyler and Gramoli, Vincent and Larrea, Mikel and Raynal, Michel},
  booktitle={2018 IEEE 17th International Symposium on Network Computing and Applications (NCA)}, 
  title={DBFT: Efficient Leaderless Byzantine Consensus and its Application to Blockchains}, 
  year={2018},
  volume={},
  number={},
  pages={1-8},
  doi={10.1109/NCA.2018.8548057}}

@techreport{alpenglow,
  title       = {Solana Alpenglow Consensus: Increased Bandwidth, Reduced Latency},
  author      = {Kniep, Quentin and Sliwinski, Jakub and Wattenhofer, Roger},
  institution = {Anza Research / ETH Zurich},
  year        = {2025},
  url         = {https://drive.google.com/file/d/1Rlr3PdHsBmPahOInP6-Pl0bMzdayltdV},
  type        = {White Paper},
}

@inproceedings{qu,
author = {Abd-El-Malek, Michael and Ganger, Gregory R. and Goodson, Garth R. and Reiter, Michael K. and Wylie, Jay J.},
title = {Fault-scalable Byzantine fault-tolerant services},
year = {2005},
isbn = {1595930795},
publisher = {Association for Computing Machinery},
address = {New York, NY, USA},
url = {https://doi.org/10.1145/1095810.1095817},
doi = {10.1145/1095810.1095817},
booktitle = {Proceedings of the Twentieth ACM Symposium on Operating Systems Principles},
pages = {59–74},
numpages = {16},
location = {Brighton, United Kingdom},
series = {SOSP '05}
}

@inproceedings{lutris,
author = {Blackshear, Sam and Chursin, Andrey and Danezis, George and Kichidis, Anastasios and Kokoris-Kogias, Lefteris and Li, Xun and Logan, Mark and Menon, Ashok and Nowacki, Todd and Sonnino, Alberto and Williams, Brandon and Zhang, Lu},
title = {Sui Lutris: A Blockchain Combining Broadcast and Consensus},
year = {2024},
isbn = {9798400706363},
publisher = {Association for Computing Machinery},
address = {New York, NY, USA},
url = {https://doi.org/10.1145/3658644.3670286},
doi = {10.1145/3658644.3670286},
booktitle = {Proceedings of the 2024 on ACM SIGSAC Conference on Computer and Communications Security},
pages = {2606–2620},
numpages = {15},
location = {Salt Lake City, UT, USA},
series = {CCS '24}
}

@misc{mysticeti,
      title={Mysticeti: Reaching the Limits of Latency with Uncertified DAGs}, 
      author={Kushal Babel and Andrey Chursin and George Danezis and Anastasios Kichidis and Lefteris Kokoris-Kogias and Arun Koshy and Alberto Sonnino and Mingwei Tian},
      year={2025},
      eprint={2310.14821},
      archivePrefix={arXiv},
      primaryClass={cs.DC},
      url={https://arxiv.org/abs/2310.14821}, 
}
\appendix
\section{Additional Related Work}

\subsection{CFT Protocols}

\par \textbf{Generalized consensus.}
Generalized consensus protocols, beginning with Generalized
Paxos~\cite{generalized-paxos} and later variants~\cite{byzantine-generalized-paxos},
also exploit the fact that non-conflicting commands need not be totally
ordered. These protocols are conceptually close to generalized leaderless SMR.
However, when conflicts are detected, they typically require a new ballot,
recovery path, or additional coordination step. This can delay even
conflict-free commands that are entangled with unresolved ordering information.
As a result, these protocols do not directly provide the leaderless execution
latency property considered here, where conflict-free commands must execute
within two total message delays even in the presence of up to $f$ Byzantine
failures.

\subsection{BFT Protocols}

\par \textbf{Single-shot agreement.}
Archipelago~\cite{rachid-leaderless} is leaderless consensus protocol in the sense that the consensus protocol can terminate without relying on one designated replica to drive progress. Our use of the term is different. We study multi-shot SMR, where clients continuously submit commands, replicas
must execute a growing history, and the protocol should avoid totally ordering commands that do not conflict. Bosco~\cite{bosco} studies single-shot Byzantine agreement with a one-message-delay fast decision path. Unfortunately, as discussed in section~\ref{sec:overview}, using a single-shot agreement protocol with a fast path in a black-box manner does not immediately satisfy Invariant 2 (Visibility).

\par \textbf{No fast path.} Asynchronous consensus protocols~\cite{honeybadger, vaba} and DAG protocols~\cite{narwhal-tusk, bullshark} have multiple proposers, but do not have optimal execution latency when there are conflict-free commands. Byblos~\cite{byblos} avoids
ordering non-conflicting commands, but its latency is always at least five message
delays. Similarly, Egalitarian BFT~\cite{ebft} does not have a fast path. DBFT~\cite{dbft} is reduces consensus to binary agreement instances that do not rely on a single leader, but it does not have a fast path that for conflict-free commands.

\par \textbf{Restricted two message delay fast path.}
Protocols such as Sui Lutris~\cite{lutris},
Mysticeti~\cite{mysticeti} have a fast path with two message delays, but this applies only to single-writer operations. In contrast,
the leaderless protocols studied in this paper are not limited to these special cases. They support two message-delay execution for any conflict-free commands.

\par \textbf{Log sharding.}
Multi-leader BFT protocols such as Mir-BFT~\cite{mir-bft},
ISS~\cite{iss}, BFT Mencius~\cite{bft-mencius}, and
RCC~\cite{rcc} improve throughput by assigning different leaders to different
log positions or partitions of the log. This reduces the bottleneck of a single
leader, but the abstraction remains fundamentally log-based. Commands are
assigned to slots, and execution generally proceeds according to the order of
those slots. Thus, even if two commands are conflict-free, a command assigned
to a later slot may have to wait for earlier slots to commit before it can be
executed. These protocols therefore do not satisfy generalized leaderless SMR
as defined in this paper, since conflict-free commands are still delayed by
unrelated earlier log positions.

\end{document}